\documentclass[letterpaper]{article} 
\usepackage{aaai2026}  
\usepackage{times}  
\usepackage{helvet}  
\usepackage{courier}  
\usepackage[hyphens]{url}  
\usepackage{graphicx} 
\usepackage{natbib}  
\usepackage{caption} 
\usepackage{algorithm}
\usepackage{algorithmic}
\usepackage{xspace}
\usepackage{threeparttable}
\usepackage{booktabs}
\usepackage{newfloat}
\usepackage{listings}
\DeclareCaptionStyle{ruled}{labelfont=normalfont,labelsep=colon,strut=off} 
\floatstyle{ruled}
\newfloat{listing}{tb}{lst}{}
\floatname{listing}{Listing}
\usepackage{amsmath}
\usepackage{amsfonts}
\usepackage{amssymb}
\usepackage{bm}
\usepackage{amsthm}
\usepackage{subcaption}

\nocopyright 
\newcommand{\kw}[1]{{\ensuremath {\mathsf{#1}}}\xspace}
\newcommand{\PriLLM}{\kw{PriLLM}}
\newcommand{\frameworkname}{\PriLLM}
\newcommand{\stitle}[1]{\noindent\textbf{#1}}
\newcommand{\etitle}[1]{\noindent\textit{#1}}
\newcommand{\mbi}{\begin{itemize}}
\newcommand{\mei}{\end{itemize}}
\newcommand{\State}{\STATE}
\newcommand{\For}{\FOR}
\newcommand{\EndFor}{\ENDFOR}
\newcommand{\If}{\IF}
\newcommand{\EndIf}{\ENDIF}
\newcommand{\Else}{\ELSE}

\def\f{{\bm f}}

\def\bb1{{\bf 1}}

\def\P{{\rm \bf P}}
\def\f{{\boldsymbol{f}}}
\def\F{{\mathcal F}}

\def\C{{\mathcal C}}
\def\P{{\mathcal P}}

\newtheorem {theorem}  {\textbf{Theorem}} 
\newtheorem {definition}  {\textbf{Definition}}
\newtheorem {problem} {\textbf{Problem}}

\newenvironment{sequation}{\begin{equation}}{\end{equation}}

\def\eg{{e.g.,}}
\def\ie{{i.e.,}}
\title{Pricing Online LLM Services with Data-Calibrated Stackelberg Routing Game}

\author{
  Zhendong Guo\textsuperscript{\rm }\protect,
  Wenchao Bai\textsuperscript{\rm },
  Jiahui Jin\textsuperscript{\rm }\thanks{Corresponding author}
}
\affiliations{
  \textsuperscript{\rm}School of Computer Science and Engineering, Southeast University\\
  \{zdguo, wbai, jjin\}@seu.edu.cn
}

\begin{document}

\maketitle

\begin{abstract}

The proliferation of Large Language Models (LLMs) has established LLM routing as a standard service delivery mechanism, where users select models based on cost, Quality of Service (QoS), among other things.
However, optimal pricing in LLM routing platforms requires precise modeling for dynamic service markets, and solving this problem in real time at scale is computationally intractable.
In this paper, we propose \PriLLM, a novel practical and scalable solution for real-time dynamic pricing in competitive LLM routing.
\PriLLM models the service market as a Stackelberg game, where providers set prices and users select services based on multiple criteria.
To capture real-world market dynamics, we incorporate both objective factors (\eg~cost, QoS) and subjective user preferences into the model.
For scalability, we employ a deep aggregation network to learn provider abstraction that preserve user-side equilibrium behavior across pricing strategies.
Moreover, \PriLLM offers interpretability by explaining its pricing decisions.
Empirical evaluation on real-world data shows that \PriLLM achieves over 95\% of the optimal profit while only requiring less than 5\% of the optimal solution's computation time.
\looseness=-1
\end{abstract}

\begin{links}
    \link{Code}{https://github.com/luck-seu/PriLLM}
\end{links}

\section{Introduction}
\label{secintro}

The landscape of Large Language Models (LLMs) is rapidly evolving, with service providers continuously introducing new models. 
This proliferation complicates the task of selecting the optimal model for users~\cite{song-etal-2025-irt, yue-etal-2025-masrouter}. 
To address this challenge, LLM routing platforms, such as OpenRouter, Eden AI, and Martian, have been developed.
These platforms provide a centralized interface that consolidates key metrics for each model, including per-token cost and Quality of Service (QoS) parameters. 
Additionally, they offer a unified API, enabling seamless access to multiple models through a standardized interface. 
{\it This paper investigates the dynamic service pricing strategies for a service provider within an LLM routing system (Figure \ref{fig:llm_routing_system})}, 
where user requests are dynamically routed to the most appropriate service provider based on user preferences. 

When selecting services, users' routing preferences are shaped by both objective and subjective factors. 
Objective attributes include per-token cost, model parameter count, 
and real-time quality-of-service metrics (QoS), namely, access delay (time to first token, TTFT) and service congestion (tokens generated per second, TPS). 
Subjective attributes encompass the user-perceived values derived from LLMs' practical utilities and their brand reputations \citep{song-etal-2025-irt,hu2024unveiling,mizrahi2024state}. 
Users aggregate the factors to minimize individual costs, allowing the routing system to allocate requests effectively across providers.
The pricing strategies in LLM service markets are crucial for both providers and users, yet remain an open problem.
On one hand, setting lower prices can attract more users, but may degrade QoS due to higher demand. On the other hand, higher prices may increase QoS but reduce market competitiveness.
Even worse, the service provider can only obtain partial information about the system
without knowing actual users' preferences and LLMs' user-perceived values.

\begin{figure}[] 
    \centering
    \includegraphics[width=\linewidth]{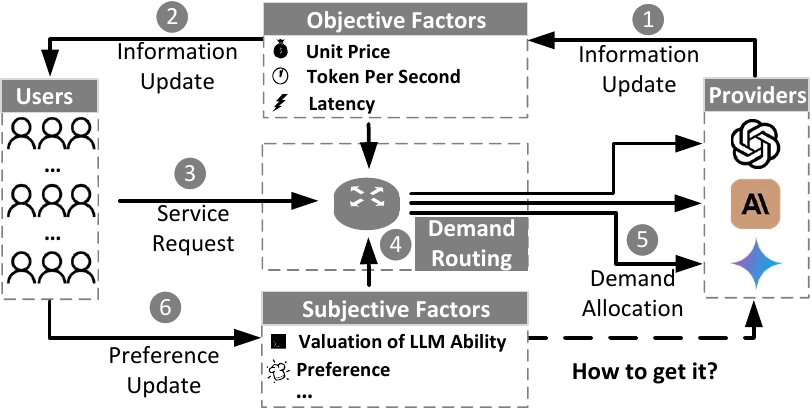} 
    \caption{
    An example of an LLM routing system that dynamically matches user requests to the most suitable providers based on individual preferences. 
    In this example, users receive updates on key metrics of service providers and then adjust their allocation strategies accordingly (Steps 1--4). 
    Their subjective preferences will be updated after providers accept and fulfill the requests (Steps 5--6).
    }
    \label{fig:llm_routing_system}
\end{figure}
Traditional approaches for cloud service pricing~\cite{saxena2024oversubscription, chakraborty2021dynamic, huang2023edgea, tutuncuoglu2024optimal,ding2023optimal} 
cannot be directly applied to LLM routing services due to a fundamental misalignment with the market's unique dynamics.
Specifically, LLM services exhibit heightened sensitivity to QoS factors such as service congestion and network latency, while user preferences are additionally shaped by user-perceived value.
As a counter-intuitive example, a provider's price increase may actually expand its market share if accompanied by substantial performance improvements.
Without a calibrated market model that jointly captures these system-level sensitivities and user-side heterogeneity, traditional methods fail to reflect real market dynamics.
The Stackelberg game framework naturally models leader-follower dynamics in service markets, where providers act as leaders setting prices and users respond as followers. 
However, applying this framework to LLM service pricing introduces substantial computational challenges.
In particular, the leader's optimization problem depends on the follower-side Nash equilibrium, which is often non-convex, leading to computational intractability.
Existing approaches mitigate this complexity through macro-level approximations~\cite{harks2021stackelberg,cui2020optimal,wu2020pool,fotouhi2021optimal,wu2021revenue,xiongoptimal} 
or equilibrium relaxations~\cite{li2019participation,naoum-sawaya2011controlled,bohnlein2021revenue,briest2008stackelberg}.
However, these simplifications sacrifice the modeling precision necessary to capture the nuanced dynamics inherent in LLM routing systems.
{\it Can we calibrate the market model to reflect real-world dynamics by incorporating both objective metrics and subjective preferences? 
If so, how can we characterize the resulting market equilibrium? 
Can we design scalable algorithms to handle many users and providers in real-time?} 
To the best of our knowledge, these questions remain unexplored.

\stitle{Contributions \& Organization.} 
To answer these questions, we propose \PriLLM, the first practical and scalable solution for real-time service pricing in LLM routing.

\etitle{(1) Problem formulation} (Section~\ref{secproblemstat}).
We formulate the service pricing problem as a Stackelberg game, where providers act as leaders and customers as followers. 
Given the initial pricing strategy profile $\mathcal{P}$, service selection strategy profile $\mathcal{F}$, user-defined routing functions, and a target provider $s$, 
it alternates between updating $\mathcal{F}$ to reach a Nash equilibrium under $\mathcal{P}$, and adjusting $s$'s pricing strategy to maximize its profit. 
This yields a mathematical program with equilibrium constraints (MPEC), which is usually NP-hard. 
We prove the existence and uniqueness of the Nash equilibrium.

\etitle{(2) Game calibration} (Section~\ref{seccalibration}).  
We calibrate our Stackelberg routing model to real LLM routing data by learning user‑preference parameters $\boldsymbol{\theta}$. 
At any given price and market condition, the user side admits a single Nash equilibrium flow, delivered by a predictor function $\mathcal{F}^{*}(\cdot)$. 
Fitting the game model to observed traffic allows us to recover latent preferences and align the model with real behaviors.

\etitle{(3) Game abstraction} (Section~\ref{secdeep_agg}).
To reduce complexity in large-scale markets, we aggregate both users and service providers into abstracted groups. 
On the user side, we group users by the apps they use, based on the assumption that users of the same app share similar preferences. 
On the provider side, we propose a deep aggregation network that learns to abstract service providers.
To prevent overfitting, \frameworkname{} samples multiple pricing strategies for $s$ and minimizes the discrepancy in selection strategies across them.

\etitle{(4) Experimental study} (Section~\ref{secnumeral_results}).
Using real-life market data from OpenRouter, we empirically find the following.
(a) \PriLLM accurately captures real-world market dynamics parameters, achieving a high $R^2$ score of 0.8982 in fitting user traffic flows. 
(b) \PriLLM outperforms pMPEC and Smooth, the best baselines, by 9.2\% and 14.3\% on average, respectively, up to 11.4\% and 17.3\%. 
(c) \PriLLM achieves over 95\% of the optimal profit, requiring less than 5\% of the computation time of the brute-force optimal solver.

\section{Related Work}
\label{secrelated_work}

\subsection{Stackelberg Game in Service Pricing}

In traditional cloud and edge computing markets, dynamic pricing has long been modeled as a Stackelberg game, with service providers acting as leaders and users acting as followers. 
However, these traditional models typically assume that users' utilities depend solely on coarse-grained quality of service (QoS) metrics~\cite{saxena2024oversubscription,chakraborty2021dynamic,tutuncuoglu2024optimal,ding2023optimal},
such as average latency, while ignoring users' subjective values of model quality.
In the LLM service domain, user decisions additionally weigh a range of fine-grained metrics, such as TTFT, TPS, and user-perceived values for different models~\cite{ding2024hybrid, hu2024routerbench, feng2024graphrouter}. 
This richer utility structure makes classical pricing models difficult to apply. 
To bridge this gap, our work introduces a Stackelberg congestion formulation explicitly parameterized by both objective and subjective preferences for the LLM servicing market.

\subsection{LLM Service Evaluation}
Existing evaluations of LLM services fall into two broad categories: 
(i) human-centric studies that elicit subjective quality judgements from small, controlled user panels \citep{chiang2024chatbot,shankar2024validates,ni2024mixeval}; 
and (ii) automated benchmarks that score models on curated reference corpora with ground-truth answers \citep{hu2024unveiling,mizrahi2024state}. 
Both streams produce leaderboards, yet neither is suitable for informing pricing decisions. 
Human-centric studies suffer from limited and non-representative samples, preventing the generalization of findings to the heterogeneous user base encountered in production systems. 
Automated benchmarks, conversely, rely on static test suites that may deviate from the task distributions encountered in live deployments; 
consequently, their reported scores exhibit low correlation with the utility users actually derive from the service. 
In contrast to these methods, we evaluate LLM performance through the lens of user-perceived value. 
This value is derived from the collective behavior of users in the LLM market and is primarily measured by the models' ability to enhance real-world task performance.

\subsection{Model Relaxation in MPEC}
A prevalent approach for investigating the optimal pricing strategy for LLM service providers within a congestion-aware Stackelberg game is to reformulate the task as an MPEC~\cite{naoum2011controlled, cardellini2016game}. 
Even a bilevel linear program is NP-hard~\cite{friesz1985properties}.
Existing research, therefore, either simplifies the problem itself~\cite{harks2021stackelberg,cui2020optimal,wu2020pool,fotouhi2021optimal,wu2021revenue,xiongoptimal} 
or relaxes the MPEC constraints~\cite{li2019participation,naoum-sawaya2011controlled,bohnlein2021revenue,briest2008stackelberg}, 
while some studies employ reinforcement learning to approximate these constraints so as to alleviate computational burden~\cite{liu2020multi, kuang2025accelerate}.
Nevertheless, users of LLM services are acutely sensitive to quality-of-service (QoS) factors such as latency and congestion. 
Consequently, model-simplification techniques that disregard heterogeneous user preferences or congestion effects are inadequate for this task~\cite{harks2021stackelberg}.
By contrast, methods that relax the MPEC constraints often require considerable computation time to attain high precision, 
whereas the complexity of the MPEC renders the underlying logic difficult for the RL network to master directly. 
Our approach is grounded in the analysis of the routing game's properties, which allows us to formulate a fully differentiable game abstraction which enables direct, end-to-end training of a neural network.

\section{Stackelberg Routing Game}
\label{secproblemstat}

\begin{figure*}[ht] 
    \centering
    \includegraphics[width=\linewidth]{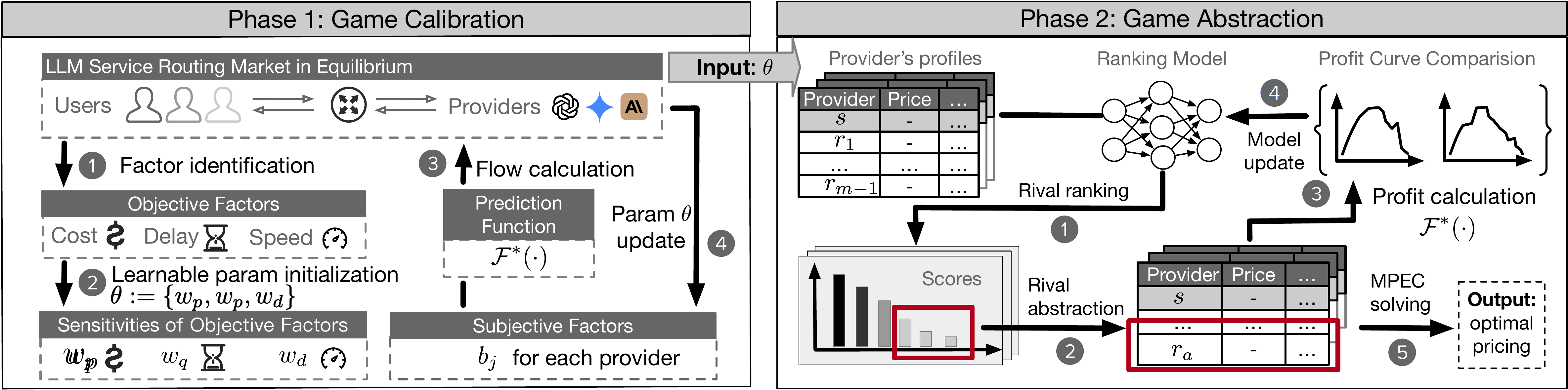} 
    \caption{Overview of \PriLLM.}
    \label{fig:arc}
\end{figure*}

We model the LLM routing system as a Stackelberg routing game. 
We begin with an overview of the game setting, followed by a formal problem definition.

\subsection{Preliminaries}

The Stackelberg routing game comprises two types of entities:  service providers and users.

\stitle{Service Providers.} The service providers develop, train, and maintain the LLMs, as well as host the LLM services.
We denote them as $\mathcal{S}$.
For each SP $s_j\in\mathcal{S}$, its service is characterized by 
user-perceived value $b_j$
, service capacity $\alpha_j$
and a certain QoS, including metrics: transmission delay between users and congestion factor.
The SP aims to set an optimal unit price $p_j$ (\eg~dollars per million tokens) to maximize its profit. 
All the SPs' prices compose a pricing strategy profile $\mathcal{P}$ of the system, such that $\mathcal{P}=\{p_j\}_{j=1}^m$.

We focus on the pricing decision problem for a (given) target provider. Thus, we divide the set of SPs into a target provider $s$ and the set of rival providers $R$, such that $\mathcal{S}=\{s\}\cup R$:
(i)~\textit{Target Provider}~($s$):~This is the specific provider whose pricing strategy is the central focus of our study. 
(ii)~\textit{Rival Providers} ($R$):~This is a set of $m-1$ competing service providers, denoted as $R = \{r_1, r_2, \dots, r_{m-1}\}$. Without loss of generality, we let $s_j\in\mathcal{S}$ be a target provider if $j=m$, otherwise $s_j$ is a rival provider.

\stitle{Users.} We consider a set of $n$ users, denoted as $U = \{u_1, u_2, \dots, u_n\}$. A user is a group of population and is aggregated as an entity, like a commercial application or an enterprise client with a total token demand $D_i$. User $u_i$ decides her allocation strategy $\f_i = \{f_{ij}\}_{j \in S}$, where $f_{ij}$ indicates the amount of tokens allocated to SP $s_j$, which is determined by  minimizing her routing-cost function $\mathcal{C}_i$:
\begin{sequation}
\label{eq:game_user}
\begin{array}{ll}
&\mathcal{C}_i  = \displaystyle \sum_{j\in S}   f_{ij}\left( w_p p_j + w_{q}  Q_j + w_{d} d_{ij}  - b_{j} \right)\\
\displaystyle & \textrm{s.t.}  \displaystyle \sum_{j\in S} f_{ij} = D_i, \quad f_{ij} \geq 0. \\
\end{array}
\end{sequation}

\noindent
where $p_j$, $d_{ij}$, and $Q_j$ represent objective factors: unit price, transmission delay between $s_j$ and $u_i$, and $s_j$'s congestion factor, respectively.
Subjective preferences are captured by $b_j$ (user-perceived value of $s_j$'s quality/brand) and $w_p, w_q, w_d$ (subjective weights of the objective factors).

\subsection{Game Formulation}
\label{sec:gameformualtion}
We use the Stackelberg game to model the pricing decision process of $s$.
In this game, $s$ as the leader sets the price $p_s$ according to the market conditions of $U$ and $R$, and then the users as followers make their allocations based on unit price, TTFT and congestion factor of each service provider. We assume $Q_j = \sum_{i \in U}f_{ij}/\alpha_j$ in our model, where $\alpha_j$ is the service capacity of $s_j$. We can get the following definitions.
The presence of service congestion means that users' selections are interdependent, as one user's choice can impact the service quality experienced by others. This interaction among users forms a non-cooperative game. Therefore, $s$'s price decision-making process is also constrained by the Nash Equilibrium (NE) of this user-side game.

\color{black}
\stitle{User-side Game} Given the fixed price strategy profile $\P$ of all service providers, users strategically split their demands to minimize their expected acquisition cost.  
Let the strategy profile of \(u_i\) be \(\f_i = \{f_{ij}\}_{j \in R \cup \{s\}}\). Denote the joint allocation strategy profile of all \(n\) users by \(\F = (\f_1, \f_2, \dots, \f_n)\). For convenience, write \(\F_{-i} = (\f_1, \dots, \f_{i-1}, \f_{i+1}, \dots, \f_n)\).
Next, we define  NE of the user-side game and prove the existence and uniqueness.

\begin{definition} [Nash Equilibrium of the user-side game]
\label{de:user_eq}
A strategy profile $\F=(\f^*_i, \f^*_{-i})$ is a Nash equilibrium of the user-side game if for any user $u_i$, it is true that $\C_i(\f^*_i; \P, \f^*_{-i}) \leq \C_i(\f'_i; \P, \f^*_{-i})$ for any $\f'_i \neq \f^*_i$.
\end{definition}

\begin{theorem}
\label{thm:ne}
A unique NE exists in the user-side game.
\end{theorem}
\begin{proof}[Proof Sketch]
The objective function of $u_i$ (\ie Eq.~\eqref{eq:game_user}) is continuous, and the inequality and equality constraints are convex. Therefore, the feasible sets of Eq.~\eqref{eq:game_user} are closed, nonempty, and convex. The Hessian matrix of the cost function $\C_{i}$ is positive definite, which means that the second-order partial derivatives are greater than zero, \ie, $\nabla^{2} \C_{i} \succ 0$ and the cost function $\C_{i}$ is strictly convex. Thus, the followers form a concave $n$-person game. By Rosen's uniqueness theorem~\cite{r6}, the lower-layer game always has a unique equilibrium, and each follower's optimization problem can converge to a unique solution in the equilibrium state, regardless of the pricing strategy $\P$.
\end{proof}

\stitle{Stackelberg Routing Game} 
Based on user-side NE, we define the target $s$'s pricing problem, which maximizes its profit $\Psi_s$ by determining the best unit price $p_s$ as below. 

\begin{problem}[LLM Service Pricing Problem] 
\label{pb:optimal_price}
Given the user-side NE, $s$ determines $p_s$ to maximize its profit $\Psi_s$.
\begin{sequation}
\begin{array}{cl}
\displaystyle \max_{p_s} &  \displaystyle \Psi_s(p_s) = p_s \cdot  Q^*_s(p_s) \cdot \alpha_s \\
\textrm{s.t.} & {\mbox{User-side game reaches equilibrium}}\\
& 0 \leq p_s \leq p^{max},
\end{array}
\end{sequation}

\noindent 
where $Q_s^{*}(p_s)\alpha_s$ is the equilibrium service demand for $s$ resulting from user-side NE. 
\end{problem}

The Stackelberg routing game (Problem~\ref{pb:optimal_price}) is intractable: its leader–follower structure, coupled with the user-side Nash equilibrium, yields an MPEC that remains NP-hard even in the linear case \cite{friesz1985properties}. Moreover, the user cost function $\mathcal{C}_i$ in Equation~\ref{eq:game_user} contains preference parameters $\boldsymbol{\theta} = \{w_p, w_q, w_d, \{b_j\}\}$ that encode latent sensitivities; unless these parameters are calibrated on real market data, the model suffers a pronounced sim-to-real gap and delivers pricing policies of inferior quality.

Accordingly, we introduce \frameworkname{}, a two-phase framework. An overview of \frameworkname{} is provided in Figure~\ref{fig:arc}: Phase~1 refines the game model by extracting latent user preferences from observational data; Phase~2 scales the market efficiently and compresses the pricing problem to a simplified form for live LLM providers.

\section{Data-Driven  Game  Calibration }
\label{seccalibration}
To calibrate the game model with the data from the real-life LLM routing platforms, we propose a paradigm shift from traditional model-first approaches to a data-driven framework. Instead of manually setting the user preference parameters $\theta$, we learn them directly from routing data. 
We consolidate these learnable parameters into a single set $\boldsymbol{\theta}$ to simplify notation, after normalizing the price weight to $w_p=1$ for model identifiability:
$\boldsymbol{\theta} = \{1,w_q, w_d, \{b_j\}_{j \in S}\}.$
We demonstrate the learnability of the parameter set $\boldsymbol{\theta}$.

To achieve game calibration, Phase 1 of \PriLLM contains four key steps. First, \PriLLM identifies objective factors from the real-world data and generates a strong initial estimate for $\theta$. 
Subsequently, \PriLLM leverages a prediction function to compute the expected routing flows under the current parameters $\boldsymbol{\theta}$. The discrepancy between these predicted flows and the actual routing data is then quantified as an error signal, which is backpropagated to refine $\boldsymbol{\theta}$. 
This refinement process is repeated until we obtain the final parameters that best explain the real-world LLM routing flows.

\subsection{Learnability of Game Parameters}
\label{secpredicte}

Theorem~\ref{thm:ne} implies that for any given $S$ and its pricing strategy $\P$, the user-side game NE is a single, routing flow $\F^*$. 
This allows us to treat the user-side game as a function $\F^*(\cdot)$ mapping  objective factors of each $s_j$, $\boldsymbol{O}=\{p_j, \alpha_j, \{d_{ij}\}_{i\in U}\}_{j \in S}$, demnad of users, $\{D_i\}_{i \in U}$ and user preference parameters, $\boldsymbol{\theta}$ to a unique equilibrium flow. 
\begin{sequation}
    \F^* = \F^{*}(\boldsymbol{\theta},\boldsymbol{O},\{D_i\}_{i \in U})
\end{sequation}

\noindent
To achieve $\F^*(\cdot)$, we construct a potential function $\Phi(\F)$ for the game. To simplify its presentation, we define it as the sum of two components: a fixed cost component $\Phi_{\text{Fixed}}(\F)$ and a congestion cost component $\Phi_{\text{Congestion}}(\F)$.
Let $\Phi_{\text{Fixed}}(\F)$ be defined as:
\begin{equation}
\label{eq:potential_fixed}
\Phi_{\text{Fixed}}(\F) = \sum_{i=1}^n \sum_{j \in S} \left( w_p p_j + w_d d_{ij} - b_j \right) f_{ij}
\end{equation}
And let $\Phi_{\text{Congestion}}(\F)$ be the total delay cost:
\begin{sequation}
\label{eq:potential_congestion}
\Phi_{\text{Congestion}}(\F) = \sum_{j \in S} \frac{w_q}{2\alpha_j} \left( (Q_j\alpha_j)^2 + \sum_{i=1}^n f_{ij}^2 \right)
\end{sequation}

\noindent
The potential function is the sum of these two parts: $\Phi(\F) = \Phi_{\text{Fixed}}(\F) + \Phi_{\text{Congestion}}(\F)$, where $\F^*$ is the uniqueness solution to minimization of $\Phi(\F;\boldsymbol{\theta}, \boldsymbol{O}, \{D_i\}_{i \in U})$.
A proof is provided in the Section~\ref{sec:ne_calculation}.
To learn $\boldsymbol{\theta}$ via gradient-based methods, we also need  compute the gradient of a loss function through $\F^{*}(\cdot)$. 
\begin{theorem}
\label{thm:optimizable}
The user-side NE prediction function $\F^*(\boldsymbol{\theta})$ is piecewise differentiable. Consequently, for any loss function $\mathcal{L}(\F^*)$, the gradient $\nabla_{\boldsymbol{\theta}} \mathcal{L}$ exists (as a sub-gradient at non-differentiable points) and can be learned end-to-end using modern automatic differentiation frameworks.
\end{theorem}
\begin{proof}[Proof Sketch]
The unique NE is the solution to a strictly convex quadratic program (QP) derived from the user-side game(a potential game), with $\boldsymbol{\theta}$ appearing as coefficients in the $\C_i$. The solution $\F^*$ is characterized by the Karush-Kuhn-Tucker (KKT) conditions, which form an implicit function of $\boldsymbol{\theta}$. Automatic differentiation libraries can differentiate through the solution of such convex optimization problems. The piecewise nature, which arises from changes in the set of active constraints (i.e., users' service routing), is handled by these frameworks, which compute valid sub-gradients at points of non-differentiability.
\end{proof}
\subsection{Initialize and Learn Game Parameters}
\label{gameparainit}
The learning process is sensitive to the initial values of $\boldsymbol{\theta}$, as poor initialization often leads to a failure to converge. Hence, we propose an initialization strategy to get the high quality initial values $\{b_j\}$ for $\boldsymbol{\theta}$ .
The method takes representative days of real-world traffic data ($\F^{\text{real}}_t$) as input, and outputs a initial parameters $\boldsymbol{\theta}_{\text{init}}$ , detailed in Section~\ref{sec:ne_calculation}.

The core idea is to assume the observed data $\F^{\text{real}}_t$ already represents a user NE. Based on this assumption, we work backward to find the inherent biases $\{b_j\}$ of the services with fixed weights $\{w_p=1,w_d=1,w_q=1\}$. In a NE, for any user, all the services they actually use must be equally costly or attractive, and these must be more attractive than any service they do not use. This principle allows us to formulate a simple Linear Program to find the smallest non-negative biases $\{b_j^*\}$ that make the real-world data calibrated with our game model. We then form our initial parameter vector:$\{w_p=1.0, w_d=1.0,w_p=1.0, \{b_j^*\}\}$.

Given initial parameter vector and $T$ days' real-world routing data as $T$ NE of user-side game, \ie, the real daily traffic distributions, objective factors of $S$, demand of each user, our objective is to find the parameters $\boldsymbol{\theta}^*$ that best account for the observed data. 
Hence we minimize a loss function that quantifies the discrepancy between the model’s predicted NE and the actual NE:

\begin{sequation}
\label{eq:inverse_opt}
\min_{\boldsymbol{\theta}} \quad \mathcal{L}(\boldsymbol{\theta}) = \sum_{t=1}^{T} \left\| \F^{*}(\boldsymbol{\theta}; \{D_{it}\}_{i \in U}, \mathbf{O}_t) - \F^{\text{real}}_t \right\|_2^2
\end{sequation}
\noindent
where $\F^{\text{real}}_t,\{D_{it}\}_{i \in U}$ is the routing data and demand data of day $t$.
By minimizing the objective function $\mathcal{L}(\boldsymbol{\theta})$, we can find a local optima using gradient-based methods.

\section{Dynamic Pricing with Game Abstraction}
\label{secdeep_agg}
\label{secstage2}
\newcommand{\modelparams}{\phi}
\newcommand{\rivalset}{R}
\newcommand{\simprrivalset}{\hat{R}_{\text{simple}}}
\newcommand{\target}{s}
\newcommand{\users}{U}
\newcommand{\candprices}{P_{\text{cand}}}
\newcommand{\predprofitvec}{\hat{\mathbf{Y}}}
\newcommand{\trueprofitvec}{\mathbf{Y}}

To tackle the Stackelberg routing game, \PriLLM~uses a novel learning framework to simplify the routing market while preserving the user-side NE for the target providers $s$. 
Within the simplified market, we can efficiently solve the MPEC problem to find a near-optimal price for target $s$.

As illustrated in Phase 2 of Fig~\ref{fig:arc}, the process is as follows. First, \PriLLM~employs a ranking model to assign a score to each rival and aggregates the rivals with lower scores.
\PriLLM~then assesses the quality of the abstracted routing game by comparing $s$'s profit curve to that of the original game, and updates the ranking model accordingly.
By using routing game abstraction, we can efficiently solve for a near-optimal pricing to the $s$'s routing game.
 
\subsection{Rival Ranking and Abstraction}
The main idea of \PriLLM's abstraction is to preserve the most influential $K-1$ rivals while aggregating the less significant ones. It can maintain the complexity of the market competition while reducing the computational load. 

To rank the rival providers, \PriLLM encodes each rival $r_j \in R$ into a feature vector encapsulating its objective factors, user-perceived value and related factor with $s$. These embeddings are processed by a stack of Set Attention Blocks (SABs)~\cite{pmlr-v97-lee19d}.
To accommodate different aggregation needs, our SAB module is designed to compute two parallel sets of importance scores for each rival:
$$    (score_j^{\text{sum}}, score_j^{\text{avg}}) = SAB(p_j,\{d_{ij}\},\alpha_j,b_j;\boldsymbol{\theta},\boldsymbol{O},\{D_i\})$$
\PriLLM preserves the top $K-1$ rivals identified by the averaging score. All remaining rivals $R_{\text{agg}}$ are then aggregated into a single representative rival, $r_a$. The attributes of $r_a$ are synthesized using the two learned score types: cumulative properties like $\alpha_a$ are calculated via a weighted sum:
$$\alpha_a = \sum_{r_j \in R_{\text{agg}}} score_j^{\text{sum}} \cdot \alpha_j$$
while attributes like $p_a,\{d_{ia}\},b_a$ are computed through a weighted average:
$$
(p_a,\{d_{ia}\},b_a) = \frac{\sum_{r_j \in R_{\text{agg}}} (score_j^{\text{avg}} \cdot (p_j,\{d_{ij}\},b_j))}{\sum_{r_j \in R_{\text{agg}}} score_j^{\text{avg}}}.
$$
This process transforms the original market into a simplified market by reducing the size of the rival set.
\subsection{Loss Function and Model Solving}

\begin{table*}[t]
\centering
\begin{threeparttable}
\centering
\begin{tabular}{l rr rr rr rr} 
    \toprule
    \textbf{Scenario} & \multicolumn{2}{c}{\textbf{Premium Market}} & \multicolumn{2}{c}{\textbf{Economy Market}} & \multicolumn{2}{c}{\textbf{Coding Market}} & \multicolumn{2}{c}{\textbf{Translation Market}} \\ 
    \cmidrule(lr){1-1} \cmidrule(lr){2-3} \cmidrule(lr){4-5} \cmidrule(lr){6-7} \cmidrule(lr){8-9} 
    \textbf{Problem Size} & \multicolumn{2}{c}{(4 LLMs, 900B tokens)} & \multicolumn{2}{c}{(8 LLMs, 1200B tokens)} & \multicolumn{2}{c}{(13 LLMs, 800B tokens)} & \multicolumn{2}{c}{(10 LLMs, 40B tokens)} \\ 
    \textbf{Metric} & \textbf{Time(s)} & \textbf{Profit(\%)} & \textbf{Time(s)} & \textbf{Profit(\%)} & \textbf{Time(s)} & \textbf{Profit(\%)} & \textbf{Time(s)} & \textbf{Profit(\%)} \\ 
    \midrule
    \textbf{\frameworkname{}}& \textbf{1.04} & \textbf{98.5\%} & \textbf{3.01} & \textbf{96.2\%} & \textbf{4.95} & \textbf{95.1\%} & \textbf{3.98} & \textbf{95.8\%} \\
    \midrule
    \kw{pMPEC}      & 1.64  & 91.9\% & 21.02 & 88.3\% & 95.04 & 85.4\% & 35.17 & 87.5\% \\
    \kw{Smooth}      & 3.29  & 88.1\% & 135.2 & 84.5\% & 680.5 & 81.1\% & 210.8 & 83.8\% \\
    \kw{SPGCE}      & 2.24  & 89.2\% & 18.04  & 85.1\% & 105.22 & 80.5\% & 32.66 & 64.2\% \\
    \kw{ODCA}      & 2.37  & 62.8\% & 28.15  & 51.3\% & 101.17 & 64.7\% & 39.49 & 78.6\% \\
    \bottomrule
\end{tabular}
\end{threeparttable}
\caption{Performance comparison of pricing algorithms with optimal profit across different market scenarios.}
\label{tab:main_results_profit}
\end{table*}

Our central hypothesis is that if the profit curve of target $s$ in an abstracted routing game aligns with the original profit curve, the abstracted routing game's derived optimal price will approximate the true optimum. 
Therefore, we evaluate the quality of an abstracted game by comparing the two profit curves of target $s$.
We employ a sampling-based method. By sampling multiple price points for target $s$, we repeatedly compute and compare its profit by  $\F^*(\cdot)$ in both the original and the abstracted routing game. 

Given the orginal profit vector $\mathbf{Y}$ and the predicted profit vector $\hat{\mathbf{Y}}(\phi)$ over a set of sampled prices points, we train \frameworkname{}, parameterized by $\phi$, by minimizing:
\begin{equation}
\mathcal{L}{\text{curve}}(\phi) = \frac{1}{L} \sum_{k=1}^{L} \left( \frac{Y_k}{||\mathbf{Y}||_{\infty}} - \frac{\hat{Y}k(\phi)}{||\mathbf{Y}||{\infty}} \right)^2,
\end{equation}

\noindent
where $\phi$ represents the parameters of our deep aggregation network.
$X$ is the total number of candidate price points sampled for the target provider $s$.
$\mathbf{Y} = [Y_1, \dots, Y_L]$ is the profit vector in which each $Y_k$ is the profit for provider $s$ at the $k$-th candidate price, computed in the original routing game.
But $\hat{\mathbf{Y}}(\phi) = [\hat{Y}_1(\phi), \dots, \hat{Y}_L(\phi)]$ is computed in the simplified game generated by \PriLLM's aggregator.
$||\mathbf{Y}||_{\infty}$ is the L-infinity norm (i.e., the maximum absolute value) of the original profit vector. We use it to normalize both curves, which stabilizes the training process against varying scales of profit. Detail calculation of NE is provided in Section~\ref{sec:ne_calculation}. 

After abstraction, the problem can be formulated as an MPEC via the KKT conditions. We can split the price domain of $p_s$ into ordered intervals, solve the resulting sub-problems, and keep the best price; As every feasible price is examined, the solution is optimal. All methods are given in Baselines of Sec 6.

\section{Evaluation}
\label{secnumeral_results}
Using real-life datasets, this section experimentally evaluate the effectiveness and efficiency of \frameworkname{} framework.

\vspace{5pt}
\noindent
{\bf Dataset.} 
We collected three months of historical data for the 20 most popular LLM services and 20 APPs from OpenRouter. 
We train \frameworkname{}'s deep aggregation network using a dataset of over 2,000 market scenarios constructed through data augmentation. Each scenario is derived from a real daily market snapshot, with perturbations applied to the attributes of rivals to simulate diverse market conditions.

\noindent
{\bf Configurations.} 
From this dataset, one provider is randomly selected as the target provider, $s$, with the rest forming the set of rivals, $R$. 
We extract the API input unit prices as $p_j$ and total daily token usage as the total user demands. To simulate a market of enterprise clients, we model the total usage flow of one commercial APP as one user. Key QoS metrics are simulated based on empirical data: user-specific TTFTs are sampled from the observed distribution on OpenRouter, and service capacities $\alpha_j=\sum_{i \in U}f_{ij} / v_j$ where $v_j$ are the models' official TPS ratings.
We report profit as a relative metric, defined as the ratio to the optimal profit.

\noindent\textbf{Enviroment settings.}
We ran experiments on a 64-bit machine with an Intel i7-8550U CPU and 24GB RAM. Our framework is implemented in Python 3.8. The MPEC problems are modeled using Pyomo 6.1.2. The learned parameters of \frameworkname{} are based on a portion of the historical data, and evaluation is performed on a held-out test set.

\begin{figure*}[htbp]
  \centering
  \includegraphics[width=\textwidth]{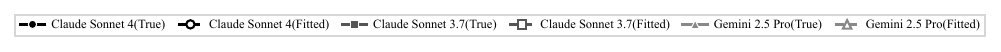}
  
  \centering
      \begin{subfigure}[b]{0.28\textwidth}
        \includegraphics[width=\textwidth]{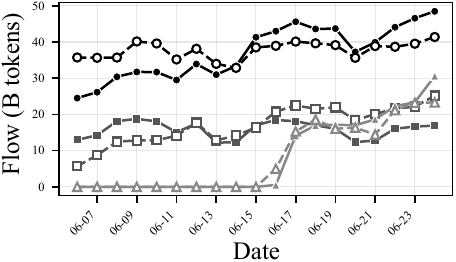} 
        \caption{Calibration of LLM Routing Flows.}
        \label{fig:model_flow}
      \end{subfigure}
      \hfill 
    \centering
    \begin{subfigure}[b]{0.17\textwidth}
      \includegraphics[width=\textwidth]{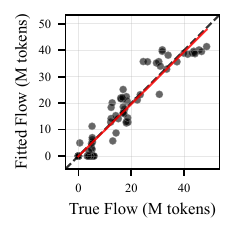} 
      \caption{Fit by \frameworkname{}. }
      \label{fig:figa}
    \end{subfigure}
    \hfill
    \begin{subfigure}[b]{0.17\textwidth}
      \includegraphics[width=\textwidth]{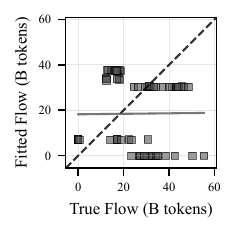} 
      \caption{Fit by \kw{XGBoost}.}
      \label{fig:figb}
    \end{subfigure}
    \hfill
    \begin{subfigure}[b]{0.17\textwidth}
      \includegraphics[width=\textwidth]{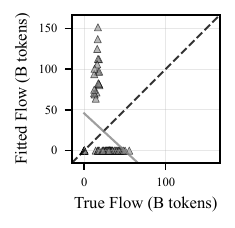} 
      \caption{Fit by \kw{NPM}.}
      \label{fig:figc}
    \end{subfigure}
    \hfill
    \begin{subfigure}[b]{0.17\textwidth}
      \includegraphics[width=\textwidth]{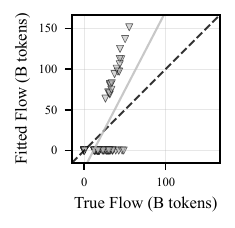} 
      \caption{Fit by $\PriLLM_{\kw{noQ}}.$}
      \label{fig:figd}
    \end{subfigure}
    \caption{Experiment of \frameworkname{}'s Game Calibration. } 
    \label{fig:group}
  \label{fig:overall}
\end{figure*}

\noindent\textbf{Baselines.}
We use the following MPEC solvers to solve the Stackelberg routing game.
(1) {\kw{SPITER}}~\cite{JIN2024110405}, the default solver integrated in \frameworkname{}.
(2) {\kw{BF} (Brute-Force)}, which enumerats all KKT conditions of the user-side game. It provides the theoretical optimum.
(3) \kw{pMPEC}~\cite{hart2017pyomo}, a generic solvers implemented via the Pyomo library.
(4) \kw{Smooth}~\cite{wu2021revenue}, which relaxes the complementarity constraints into smooth inequalities, allowing the use of standard nonlinear solvers.
(5) \kw{ODCA}~\cite{chen2020stackelberg} \&  \kw{SPGCE}~\cite{harks2021stackelberg}, a simplified MPEC solver by ignoring user interaction effects.

For market simulation, we include the following baselines: 
(6) \kw{NPM}, which models the user cost function solely based on observable, objective metrics.
(7) \kw{XGBoost}~\cite{chen2016xgboost}, a non-game-theoretic approach that directly predicts user choices from market features.

We also compare with baselines for subjective preference learning:
(8) \kw{FD}, a black-box optimization method that approximates the gradients of the user equilibrium through finite differences.  
(9) \kw{A2C}~\cite{liu2020multi}, an actor-critic reinforcement learning method that learns an optimal policy for selecting preference parameters.

\noindent
\subsection{Overall Performance of \frameworkname{}}
\label{secbehavior_calibration}

We conduct experiments on \frameworkname{} to validate the game calibration and the game abstraction method. To evaluate \frameworkname{} under diverse competitive conditions, we formulated the LLM market into four typical scenarios: Premium and Economy markets are delineated by API price points (e.g., above/below \$1 per million tokens), while Coding and Translation markets are formulated based on application-specific usage data. This partitioning varies problem sizes and competitive dynamics, as summarized in Table~\ref{tab:main_results_profit}.

\paragraph{Effectiveness of game calibration.} 
We tested \frameworkname{} on the coding market which includes 13 popular LLM service providers and 8 coding apps with token usage exceeding 2 B tokens. 
Fig.~\ref{fig:model_flow} shows that \frameworkname{} achieves good fitting performance in both the 2-LLM Model and 3-LLM Model coding market scenarios based on the learned parameters. 
Figure~\ref{fig:figa} show that our model achieves a high $R^2$ score of 0.8982 and a low Mean Absolute Error (MAE) of 3.56M tokens. Comparing with Fig.~\ref{fig:figb}, we conclude that \frameworkname{} has stronger prior knowledge than traditional machine learning methods and can learn effective information when relevant data is sparse. Comparing with Fig.~\ref{fig:figc} and Fig.~\ref{fig:figd}, we see that the $b_j$ and congestion effect $Q_j$ considered in the modeling of \frameworkname{} improve the 
ability of calibrating.

\paragraph{Effectiveness of game abstraction.} We validate \frameworkname{}'s ability to simplify the Stackelberg routing game. 
We use the deep aggregation network to get a simplified market with $K=2$ rivals, for which $K=2$ is a good choice in terms of solution efficiency and aggregation accuracy. Then We compare \frameworkname{}'s efficiency and accuracy with different methods on the various market settings. The results are mainly represented in Table~\ref{tab:main_results_profit}. Comparing \frameworkname{} with \kw{pMPEC} and \kw{Smooth} methods, we can find that according to the aggregation results obtained by \frameworkname{}, the \kw{SPITER} method can solve the problem in a faster time and obtain profit close to the optimal solution. 

\noindent
\subsection{Impacts of Key Components}
\paragraph{Impact of game-parameter learning method.}
We study the impact of the game-parameter learning methods on game calibration.
{First, we collected data from the four most popular LLMs for 19 consecutive days, including price, TTFT, TPS, and usage tokens from various apps. We then used  this data to learn user subjective parameters in our game and then calibrate them against the real market. Parameters learned using the \kw{A2C} or \kw{FD} methods consistently failed to effectively calibrate the model to real-world data. Due to the complexity of the underlying user game, \kw{A2C} methods struggle to learn. \kw{FD} methods treat the underlying user solution process as a black box, requiring a significant time-consuming gradient calculation and manual adjustment of the update step size
}
Table~\ref{tab:ablation_study} shows that only our method can stably learn parameters. On the dataset, the MSE error can be reduced by 45.73 within 25.88 seconds.

\begin{table}[h]
\centering
\begin{tabular}{l c r r}
\toprule
\textbf{Method} & \textbf{Converges?} & \textbf{Final MSE} & \textbf{Time (s)} \\
\midrule
\textbf{\frameworkname{}} & \textbf{Yes} & \textbf{45.73} & \textbf{25.88} \\
\kw{FD}       & No & $1.63 \times 10^6$  & --- \\
\kw{A2C}              & No & $3.14 \times 10^{11}$ & --- \\
\bottomrule
\end{tabular}
\caption{Impact of Game-parameter Learning Methods.}
\label{tab:ablation_study}
\end{table}
\paragraph{Impact of aggregation methods.} We evaluate the impacts of aggregation methods on game abstraction. 
First, we selected data from eight LLMs in the Economy market as the baseline experimental setup, including unit price, TTFT, TPS, and usage tokens from various apps. Second, we varied the number of LLMs from small to large (5 and 7 rival providers). Finally, based on the resulting markets of varying sizes, we simplified the market using different aggregation methods and uniformly calculated the optimal price using the \kw{SPITER} algorithm. In Table~\ref{tabablation_da_net}, we compared the profits obtained using different aggregation methods. DA$_{K=2}$ 
We also compared the time taken to calculate the price using the \kw{BF} method directly without using any aggregation method.
In Table~\ref{tabablation_da_net}, we can see that \frameworkname{} achieves near-optimal profits (over 97\% of the optimum) while drastically reducing computation time compared to the exact solver.

\begin{table}[ht]
\centering
\caption{
Impact of Aggregation Methods. \kw{BF} is the brute-force method without aggregation. \kw{DA} stands for the deep aggregation network. \kw{MIN} selects the $K$ cheapest rivals. \kw{AVG} returns $K$ identical providers with averaged attributes.
}
\label{tabablation_da_net}
\begin{tabular}{l rr rr}
    \toprule
    & \multicolumn{2}{c}{\textbf{6 LLMs Market  }} & \multicolumn{2}{c}{\textbf{8 LLMs Market}} \\
    \cmidrule(lr){2-3} \cmidrule(lr){4-5}
    \textbf{Method} & \textbf{Profit(\%)} & \textbf{Time(s)} & \textbf{Profit(\%)} & \textbf{Time(s)} \\
    \midrule
    \kw{BF} & 100.0\% & 420.63 & 100.0\% & 612.24 \\
    \midrule
    \kw{DA}$_{K=1}$ & 93.8\% & 5.52 & 92.1\% & 9.85 \\
    \kw{DA}$_{K=2}$ & 95.2\% & 9.18 & 94.5\% & 13.24 \\
    \kw{DA}$_{K=3}$ & 99.7\% & 19.17 & 99.5\% & 20.21 \\
    \kw{DA}$_{K=4}$ & 99.9\% & 62.65 & 99.9\% & 61.22 \\
    \midrule
    \kw{MIN}$_{K=2}$ & 58.1\% & 6.37 & 51.7\% & 14.09 \\
    \kw{AVG}      & 18.9\% & 5.15 & 15.3\% & 9.98 \\
    \bottomrule
\end{tabular}
\end{table}
\paragraph{Impact of Parameter $K$.} We also analyze the impact of the number of aggregated rivals $K$ on the model's effectiveness and efficiency. 
We conduct this experiment on the 8 LLM market settings and compare with heuristic method. We vary $K$ from 1 to 4 for \frameworkname{} and measure the resulting profit and total computation time. Table~\ref{tabablation_da_net} shows that increasing $K$ from 1 to 4 reduces the profit gap from 11.53\% to a near-zero 0.06\%. Our experiment shows that using $K=2$ aggregated rivals offers a profit improvement over $K=1$ for a increase in runtime.

\section{Conclusion}
\label{secconclusion}
We introduced \PriLLM, a framework for dynamic pricing in LLM service markets. By formulating the problem as a Stackelberg game and leveraging novel data-driven calibration and a deep aggregation network, \PriLLM overcomes the limitations of traditional approaches. Our evaluation on real-world data confirms that \PriLLM achieves near-optimal profit with efficiency. As future work, we will expand \frameworkname{} from single leader pricing to multiple leaders.

\section{Acknowledgments}
This work is supported by National Natural Science Foundation of China under Grants No. 62572119 and 62232004,  Jiangsu Provincial Key Laboratory of Network and Information Security under Grants No.BM2003201, Key Laboratory of Computer Network and Information Integration of Ministry of Education of China under Grants No.93K-9, and partially supported by Collaborative Innovation Center of Novel Software Technology and Industrialization, Collaborative Innovation Center of Wireless Communications Technology. We also thank the Big Data Computing Center of Southeast University for providing the experiment environment and computing facility.

\bibliography{aaai2026}

\clearpage
\section{Appendix}
\label{sec:appendix}
\begin{table}[ht] 
    \centering 
    \caption{Key Notations in the Stackelberg Routing Game.} 
    \label{tab:notations} 
    \begin{tabular}{ll}
        \toprule 
        $\mathcal{S}$ & Set of all service providers (SPs). \\
        $U$ & Set of all users. \\
        $s$ & The target SP. \\
        $R$ & Set of rival SPs. \\
        $s_j, u_i$ & The $j$-th SP and the $i$-th user, respectively. \\
        $m, n$ & Total number of SPs and users, respectively. \\
        $b_j$ & User-perceived value for SP $s_j$. \\
        $\alpha_j$ & Service capacity of SP $s_j$. \\
        $D_i$ & Total token demand of user $u_i$. \\
        $d_{ij}$ & Transmission delay between $u_i$ and $s_j$. \\
        $\boldsymbol{\theta}$ & Vector of all user preference parameters. \\
        $p^{\max}$ & Maximum allowable price for the target SP. \\
        $p_j$ & Unit price set by SP $s_j$. \\
        $\mathcal{P}$ & Price profile of all SPs in the system. \\
        $f_{ij}$ & Amount of tokens allocated from $u_i$ to $s_j$. \\
        $\mathbf{f}_i$ & Allocation strategy vector for user $u_i$. \\
        $\mathbf{F}$ & Joint allocation strategy profile of all users. \\
        $\mathcal{C}_i$ & Cost function for user $u_i$. \\
        $\Psi_s$ & Profit function for the target SP $s$. \\
        $Q_j$ & Congestion factor of SP $s_j$. \\ 
        $\Phi(\mathbf{F})$ & Exact potential function for the user-side game. \\
        $\mathbf{F}^*$ & Nash Equilibrium (NE) of the user-side game. \\ 
        $\mathbf{F}^{\text{real}}$ & Observed real-world flow allocation from data. \\
        \bottomrule 
    \end{tabular}
\end{table}

\subsection{Data Acquisition and Parameter Calibration}
\label{sec:data_acquisition}

We source our dataset from OpenRouter. For example, the dataset of programming market comprises 13 LLM service providers ($s_j \in \mathcal{S}$) and 11 major applications ($u_i \in U$). For each service provider $s_j$, we collect its unit price $p_j$, average latency, and generation speed in tokens per second (TPS), which we denote as $v_j$. The total weekly token usage for a provider is denoted as $T_j$. We then estimate its service capacity as $\alpha_j = T_j / v_j$. The traffic flow from each application $u_i$ to each provider $s_j$, representing the observed flow $f_{ij, t}^{\text{real}}$, is extracted from the top public apps this week using this model section on each provider's details page. To ensure data robustness, we filter out any application whose usage constitutes less than 1\% of the top application's volume for a given model. The transmission delay $d_{ij}$ is approximated using the average service latency reported by OpenRouter. Table~\ref{tab:app_usage_dataset} and~\ref{tab:model_performance_dataset} summarize the structure and key fields of our datasets.

We adopt the core assumption that the observed weekly routing flow constitutes a user-side Nash Equilibrium. With this equilibrium data, and by leveraging the established piecewise differentiability of the equilibrium mapping and the learnability proof in Theorem~\ref{thm:optimizable}, we can calibrate the model parameters. We employ a gradient-based optimization method to learn the user preference weights $w_q, w_d$ and the provider-specific perceived values $\{b_j\}$. 

\begin{table}[H]
\centering
\caption{Description of the APP usage dataset. The data spans from April 13, 2025, to July 23, 2025, covering the top 20 most active APPs' interactions with the top 20 LLMs.}
\label{tab:app_usage_dataset}
\begin{tabular}{@{}p{0.40\linewidth}p{0.64\linewidth}@{}}
\toprule
\textbf{Field Name} & \textbf{Description} \\
\midrule
Date & The date of data recording in YYYY-MM-DD format. \\
app\_name & The unique identifier or name of the application making the API call. \\
model\_name & The specific language model being called by the application. \\
model\_usage\_token & The total number of tokens consumed by this app for this model (raw count). \\
output\_speed & The average generation speed experienced by this app, in tokens per second (tokens/s). \\
time\_to\_first\_token & The average time to first token experienced by this app, in seconds (s). \\
\bottomrule
\end{tabular}
\end{table}

\begin{table}[H]
\centering
\caption{Description of the LLM performance dataset. The data spans from July 5, 2025, to July 23, 2025, and covers the top 20 most-used LLMs.}
\label{tab:model_performance_dataset}
\begin{tabular}{@{}p{0.40\linewidth}p{0.64\linewidth}@{}}
\toprule
\textbf{Field Name} & \textbf{Description} \\
\midrule
Date & The date of data recording in YYYY-MM-DD format. \\
model\_name & The unique identifier or name of the language model. \\
total\_token\_usage\_M & The total token consumption for the model on a given day, in millions. \\
output\_speed & The average generation speed of the model, measured in tokens per second (tokens/s). \\
time\_to\_first\_token & The average time to generate the first token, measured in seconds (s). \\
\bottomrule
\end{tabular}
\end{table}

\subsection{User-Side Equilibrium Existence and Uniqueness}
\label{useq}

We provide detailed proof for Theorem \ref{thm:ne} in Section~\ref{sec:gameformualtion}.

\begin{proof}
We prove the existence and uniqueness of the Nash Equilibrium by verifying the conditions of Rosen's theorem for the uniqueness of equilibrium points in $n$-person games \cite{r6}. Our user-side game is a cost-minimization game, which is equivalent to a payoff-maximization game where the payoff function is the negative of the cost function.

First, we analyze the strategy space for each user $u_i \in U$. A user's strategy $\f_i = \{f_{ij}\}_{j \in \mathcal{S}}$ is constrained by $\sum_{j \in \mathcal{S}} f_{ij} = D_i$ and $f_{ij} \geq 0$ for all $j \in \mathcal{S}$. This feasible strategy set for user $u_i$ is a standard simplex, which is a non-empty, compact, and convex subset of $\mathbb{R}^{|\mathcal{S}|}$.

Second, we examine the user's cost function $\mathcal{C}_i$ from Eq.~\eqref{eq:game_user}. This function is continuous in the joint strategy profile $\F$. For a fixed strategy profile $\f_{-i}$ of other users, we demonstrate that $\mathcal{C}_i$ is a strictly convex function of user $u_i$'s own strategy, $\f_i$.

The Hessian matrix of $\mathcal{C}_i$ with respect to the variables in $\f_i$ is a diagonal matrix. Its diagonal elements are given by:
\begin{equation}
\frac{\partial^2 \mathcal{C}_i}{\partial f_{ij}^2} = \frac{2w_q}{\alpha_j}
\end{equation}
and all off-diagonal elements are zero. Given that the weight $w_q > 0$ and service capacity $\alpha_j > 0$, these diagonal elements are strictly positive. Consequently, the Hessian matrix $\nabla^2_{\f_i} \mathcal{C}_i$ is positive definite. This proves that for any fixed strategies of other users $\f_{-i}$, the cost function $\mathcal{C}_i$ is strictly convex with respect to $\f_i$.

A game in which each player minimizes a strictly convex function over a compact, convex set is known as a convex game. This is equivalent to a concave game where each player maximizes a strictly concave payoff function $(-\mathcal{C}_i)$. Such a game satisfies the conditions of Rosen's theorem, which guarantees the existence and uniqueness of a Nash equilibrium. Therefore, the user-side game admits a unique equilibrium strategy profile $\F^*$, regardless of the providers' pricing strategy profile $\mathcal{P}$.
\end{proof}

\subsection{Users' NE Calculation}
\label{sec:ne_calculation}

We provide calculation method of user-side NE for Section~\ref{secpredicte}.

The existence and uniqueness of the Nash Equilibrium (NE) for the user-side game are established by demonstrating that it is a convex game, which is equivalent to a potential game. The NE strategy profile $\F^*$ can be computed by finding the unique minimizer of an exact potential function $\Phi(\F)$ over the feasible set of user strategies. This transforms the multi-agent equilibrium problem into a single, tractable convex optimization problem.

First, we define the exact potential function $\Phi(\F)$ for the user-side game, as introduced in the proof. The function is composed of two parts: a fixed cost component and a congestion cost component. Let $\Phi_{\text{Fixed}}(\F)$ be defined as:
\begin{equation}
\label{eq:potential_fixed_calc}
\Phi_{\text{Fixed}}(\F) = \sum_{i=1}^n \sum_{j \in \mathcal{S}} \left( w_p p_j + w_d d_{ij} - b_j \right) f_{ij}
\end{equation}
And let $\Phi_{\text{Congestion}}(\F)$ be the congestion-related component:
\begin{sequation}
\label{eq:potential_congestion}
\Phi_{\text{Congestion}}(\F) = \sum_{j \in S} \frac{w_q}{2\alpha_j} \left( (Q_j\alpha_j)^2 + \sum_{i=1}^n f_{ij}^2 \right)
\end{sequation}
The complete potential function $\Phi(\F)$ is the sum of these two parts:
\begin{equation}
\label{eq:atomic_potential_combined_calc}
\Phi(\F) = \Phi_{\text{Fixed}}(\F) + \Phi_{\text{Congestion}}(\F)
\end{equation}
It can be verified that the partial derivative of this potential function with respect to a user's flow, $\frac{\partial \Phi(\F)}{\partial f_{ij}}$, is equal to:
\begin{align}
\frac{\partial \Phi(\F)}{\partial f_{ij}} = \frac{\partial \C}{\partial f_{ij}}
\end{align}
The equilibrium is found by minimizing $\Phi(\F)$.

The unique NE strategy profile $\F^*$ is the solution to the following convex optimization problem:
\begin{sequation}
\label{eq:ne_optimization_problem}
\begin{array}{ll}
\displaystyle \min_{\F} & \Phi(\F) \\
\textrm{s.t.} & \displaystyle \sum_{j \in \mathcal{S}} f_{ij} = D_i, \quad \forall i \in U, \\
& f_{ij} \geq 0, \quad \forall i \in U, j \in \mathcal{S}. \\
\end{array}
\end{sequation}

\noindent
Since the objective function $\Phi(\F)$ is strictly convex and the feasible strategy space is a non-empty, compact, and convex set, a unique solution $\F^*$ exists and corresponds to the NE of the user-side game.

The solution to this optimization problem is characterized by the Karush-Kuhn-Tucker (KKT) conditions. Let $\lambda_i$ be the Lagrange multiplier for the demand constraint of user $u_i$, and $\mu_{ij}$ be the multiplier for the non-negativity constraint on $f_{ij}$. The KKT conditions for the equilibrium $\F^*$ are:
\begin{itemize}
    \item \textbf{Stationarity:} For every user $u_i \in U$ and provider $s_j \in \mathcal{S}$:
    \begin{equation}
        \frac{\partial \Phi(\F^*)}{\partial f_{ij}} + \lambda_i - \mu_{ij} = 0
    \end{equation}
    \item \textbf{Primal Feasibility:}
    \begin{equation}
        \sum_{j \in \mathcal{S}} f^*_{ij} = D_i, \quad f^*_{ij} \geq 0
    \end{equation}
    \item \textbf{Dual Feasibility:}
    \begin{equation}
        \mu_{ij} \geq 0
    \end{equation}
    \item \textbf{Complementary Slackness:}
    \begin{equation}
        \mu_{ij} f^*_{ij} = 0
    \end{equation}
\end{itemize}

These conditions provide the logic for the equilibrium allocation. From the complementary slackness condition, if user $u_i$ allocates a positive flow to provider $s_j$ (i.e., $f^*_{ij} > 0$), then $\mu_{ij}$ must be zero. The stationarity condition then simplifies to $\frac{\partial \Phi(\F^*)}{\partial f_{ij}} = -\lambda_i$. This implies that for any given user $u_i$, the marginal potential cost must be equal for all service providers $s_j$ to which she allocates positive demand. For any provider $s_k$ not used by $u_i$ (i.e., $f^*_{ik}=0$), we have $\mu_{ik} \geq 0$, which implies $\frac{\partial \Phi(\F^*)}{\partial f_{ik}} \geq -\lambda_i$.

In summary, at equilibrium, each user distributes their demand $D_i$ among the service providers such that the marginal costs on all chosen routes are equal, and this marginal cost is less than or equal to the marginal cost on any unchosen route. This is the classic Wardrop's first principle for user equilibrium, and problem~\eqref{eq:ne_optimization_problem} provides a direct method for its computation.

\subsection{Learnability of Game Parameters}

We provide detailed proof for Theorem~\ref{thm:optimizable} in Section~\ref{secpredicte}.

\begin{proof}

The user-side game admits a unique Nash Equilibrium $\F^*$, which is the solution to the minimization of an exact potential function $\Phi(\F)$. We express this dependency on the model parameters $\boldsymbol{\theta}$ by writing the equilibrium as a function $\F^*(\boldsymbol{\theta})$. This function is the unique solution to the following strictly convex quadratic program (QP):
\begin{equation}
\label{eq:potential_minimization_theta}
\F^*(\boldsymbol{\theta}) = \arg\min_{\F \in \mathcal{F}} \Phi(\F; \boldsymbol{\theta}),
\end{equation}
where $\mathcal{F} = \{ \F \mid \forall u_i \in U, \sum_{j \in \mathcal{S}} f_{ij}=D_i, f_{ij} \geq 0 \}$ is the convex and compact set of feasible flow allocations.

The potential function $\Phi(\F; \boldsymbol{\theta})$ for this atomic-splittable congestion game is a polynomial in the variables $f_{ij}$ and the parameters $w_p, w_q, w_d, b_j$. Therefore, $\Phi(\F; \boldsymbol{\theta})$ is twice continuously differentiable ($C^2$) with respect to both $\F$ and $\boldsymbol{\theta}$.

The unique solution to the QP in Eq.~\eqref{eq:potential_minimization_theta} is characterized by its KKT conditions. Let $\bm{\lambda} \in \mathbb{R}^n$ be the vector of Lagrange multipliers for the demand equality constraints ($\sum_j f_{ij} = D_i$) and $\bm{\mu} \in \mathbb{R}^{n \times |\mathcal{S}|}$ be the multipliers for the non-negativity constraints ($f_{ij} \ge 0$). The KKT system is:
\begin{subequations}
\label{eq:kkt_system_theta}
\begin{align}
\nabla_{\F} \Phi(\F; \boldsymbol{\theta}) + A^\top \bm{\lambda} + \bm{\mu} &= \bm{0} \label{eq:kkt_stationarity_theta} \\
A\F - \mathbf{D} &= \bm{0} \label{eq:kkt_eq_feasibility_theta} \\
f_{ij} \ge 0, \quad \mu_{ij} \ge 0, \quad \mu_{ij} f_{ij} &= 0 \quad &\forall i \in U, j \in \mathcal{S} \label{eq:kkt_complementarity_theta}
\end{align}
\end{subequations}
where $A$ is the matrix representing the linear equality constraints and $\mathbf{D}$ is the vector of demands $\{D_i\}_{i=1}^n$.

We invoke the Implicit Function Theorem on the system~\eqref{eq:kkt_system_theta} to show that the solution $(\F^*, \bm{\lambda}^*, \bm{\mu}^*)$ is a differentiable function of $\boldsymbol{\theta}$. The theorem requires the Jacobian of the system with respect to the primal-dual variables $(\F, \bm{\lambda}, \bm{\mu})$ to be non-singular at the solution.

Consider a solution point $(\F^*, \bm{\lambda}^*, \bm{\mu}^*)$ that is non-degenerate, meaning it satisfies strict complementarity: for each $(i, j)$, either $f^*_{ij} > 0$ and $\mu^*_{ij} = 0$, or $f^*_{ij} = 0$ and $\mu^*_{ij} > 0$. At such points, the active constraint set is stable under small perturbations of $\boldsymbol{\theta}$. The Jacobian of the active part of the KKT system with respect to $(\F, \bm{\lambda})$ is the KKT matrix:
\[
\mathcal{J} = 
\begin{bmatrix}
\nabla^2_{\F\F} \Phi(\F^*; \boldsymbol{\theta}) & A^\top \\
A & \bm{0}
\end{bmatrix}
\]
The proof for the uniqueness of the NE established that the potential function $\Phi(\F)$ is strictly convex. This implies that its Hessian, $\nabla^2_{\F\F} \Phi$, is positive definite. The constraints defined by matrix $A$ are linear and satisfy the Linear Independence Constraint Qualification (LICQ). For a strictly convex program under LICQ, the KKT matrix $\mathcal{J}$ is non-singular.

Since the functions defining the KKT system~\eqref{eq:kkt_system_theta} are continuously differentiable in both the variables and the parameters $\boldsymbol{\theta}$, and the Jacobian $\mathcal{J}$ is invertible, the Implicit Function Theorem applies. It guarantees that the solution map $(\F^*, \bm{\lambda}^*, \bm{\mu}^*)$ is a locally unique and continuously differentiable function of the parameters $\boldsymbol{\theta}$ in a neighborhood of any non-degenerate point.

The points in the parameter space of $\boldsymbol{\theta}$ where differentiability might fail are the degenerate points where the active constraint set changes (i.e., some $f^*_{ij}$ or $\mu^*_{ij}$ transitions to or from zero). These points form a set of measure zero. Therefore, we conclude that the equilibrium mapping $\F^*(\boldsymbol{\theta})$ is piecewise differentiable with respect to the parameters $w_p, w_q, w_d,$ and $\{b_j\}$.
Now, we have established that the unique Nash Equilibrium allocation, $\F^*(\boldsymbol{\theta})$, is a piecewise differentiable function of the model parameters $\boldsymbol{\theta}$. To compute the gradient of the loss, we apply the chain rule:
\begin{equation}
\nabla_{\boldsymbol{\theta}} \mathcal{L}(\F^*(\boldsymbol{\theta})) = \left(\frac{\partial \F^*(\boldsymbol{\theta})}{\partial \boldsymbol{\theta}}\right)^\top \nabla_{\F^*} \mathcal{L}(\F^*)
\end{equation}
where $\frac{\partial \F^*(\boldsymbol{\theta})}{\partial \boldsymbol{\theta}}$ is the Jacobian of the equilibrium mapping and $\nabla_{\F^*} \mathcal{L}$ is the gradient of the loss with respect to the allocation.

We analyze the two components of this product:
\mbi
    \item Gradient of the Loss $\nabla_{\F^*} \mathcal{L}$: The loss function $\mathcal{L}$ is a design choice, typically selected to be a differentiable function such as Mean Squared Error. Therefore, its gradient with respect to its input $\F^*$ is well-defined and readily computable.

    \item Jacobian of the Equilibrium Mapping $\frac{\partial \F^*(\boldsymbol{\theta})}{\partial \boldsymbol{\theta}}$: As proven, the mapping $\boldsymbol{\theta} \mapsto \F^*(\boldsymbol{\theta})$ is differentiable everywhere except on a set of measure zero. This set corresponds to parameter values where the active constraint set of the underlying QP (from Eq.~\eqref{eq:potential_minimization_theta}) changes. At all points of differentiability, the Jacobian exists.
\mei

At the points of non-differentiability, the standard gradient does not exist. However, because the equilibrium is the solution to a convex optimization problem, the function $\F^*(\boldsymbol{\theta})$ is continuous, and the loss function $\mathcal{L}(\F^*(\boldsymbol{\theta}))$ admits sub-gradients. For the purpose of stochastic gradient-based optimization, a sub-gradient provides a valid descent direction.

Modern automatic differentiation frameworks are equipped to handle such scenarios. They can perform implicit differentiation by leveraging the KKT conditions (Eq.~\eqref{eq:kkt_system_theta}) that implicitly define $\F^*$ as a function of $\boldsymbol{\theta}$. These frameworks can differentiate through the solution of the convex QP. When they encounter a point of non-differentiability, they return a valid sub-gradient (e.g., a one-sided derivative), which is sufficient for optimization algorithms like SGD or Adam to converge.

Therefore, since the gradient (or a sub-gradient) of the loss function $\mathcal{L}$ with respect to $\boldsymbol{\theta}$ can be computed for all values of $\boldsymbol{\theta}$, the parameters are learnable. We can effectively train the model end-to-end by backpropagating through the equilibrium-finding process to update $\boldsymbol{\theta}$.
\end{proof}

\subsection{Initialization of Subjective Preferences $\{b_j\}$}
\label{sec:bias_initialization}

We provide the initialization method of subjective preferences $\{b_j\}$ in Section~\ref{gameparainit}.

The KKT conditions for the equilibrium $\F^*$ imply that for any user $u_i$, the marginal potential cost must be equal for all services $s_j$ to which she allocates positive flow ($f^*_{ij} > 0$). This equilibrium marginal cost must be less than or equal to the marginal potential cost for any service $s_k$ that she does not use ($f^*_{ik} = 0$). Let us denote the equilibrium marginal cost for user $u_i$ by an auxiliary variable $\nu_i$. To initialize the search for $\{b_j\}$, we fix the weight parameters to baseline values, e.g., $w_p=1.0$, $w_q=1.0$, and $w_d=1.0$. Let $M'_{ij,t}$ denote the observable component of the marginal cost for user $i$ on service $j$ at day $t$, calculated using the real data equalibrium $\F^{\text{real}}_t$:
\begin{equation}
\label{eq:marginal_cost_observable}
M'_{ij,t} = w_p p_{j,t} + w_d d_{ij,t} + w_q \left( \frac{1}{\alpha_j} \sum_{k=1}^n f_{kj,t}^{\text{real}} + \frac{1}{\alpha_j} f_{ij,t}^{\text{real}} \right) 
\end{equation}
Then, the KKT conditions for the observed equilibrium $\F^{\text{real}}_t$ translate into a set of linear constraints on the unknown parameters $\{b_j\}$ and $\{\nu_i\}$. To select a unique and parsimonious solution, we seek the smallest non-negative biases $\{b_j\}$ that satisfy these constraints by solving the following Linear Program (LP):

\label{pb:bias_init}
\begin{align}
\min_{\{b_j\}, \{\lambda_i\}} & \quad \sum_{j \in \mathcal{S}} b_j \\ 
\text{s.t.} & \quad M'_{ij,t} - b_j = \lambda_i, \quad \forall i \in U, j \in \mathcal{S} \text{ s.t. } f_{ij, t}^{\text{real}} > 0 \label{eq:init_cond1} \\ 
& \quad M'_{ij,t} - b_j \ge \lambda_i, \quad \forall i \in U, j \in \mathcal{S} \text{ s.t. } f_{ij, t}^{\text{real}} = 0 \label{eq:init_cond2} \\ 
& \quad b_j \ge 0, \quad \forall j \in \mathcal{S} \label{eq:init_cond3}
\end{align}
\noindent As formulated, Problem~\ref{pb:bias_init} is a standard LP that can be solved efficiently to obtain the optimal biases $\{b_j^*\}$. 
The initial parameter vector for the calibration is therefore set to $\boldsymbol{\theta}_{\text{init}} = \{w_p=1.0, w_q=1.0, w_d=1.0, \{b_j^*\}\}$.

\begin{algorithm}[H]
\caption{Initialization of Preference Parameters $\{b_j\}$}
\label{alg:init}
\begin{algorithmic}[1]
\REQUIRE Observed flow $\mathbf{F}^{\text{real}}$, prices $\mathbf{P}$, latencies $\mathbf{d}$, capacities $\bm{\alpha}$
\State Set initial weights $w_p \gets 1.0, w_d \gets 1.0, w_q \gets 1.0$ 
\State Calculate observed congestion $Q_j \gets \sum_i f^{\text{real}}_{ij} / \alpha_j$ for all $j$
\State Define LP variables: $b_j$ for each provider $j$, $\lambda_i$ for each user $i$
\State Define objective: $\min \sum_j b_j$
\State Add constraints based on KKT conditions:
\For{each user $i$ and provider $j$}
    \State $C_{ij} \gets w_p p_j + w_d d_{ij} + w_q (Q_j + f^{\text{real}}_{ij} / \alpha_j) - b_j$ 
    \If{$f^{\text{real}}_{ij} > 0$}
        \State Add constraint $C_{ij} = \lambda_i$
    \Else
        \State Add constraint $C_{ij} \ge \lambda_i$
    \EndIf
\EndFor
\State Add constraint $b_j \ge 0$ for all $j$
\State Solve the LP to find optimal $b_j^*$
\State \RETURN $b_j^*$
\end{algorithmic}
\end{algorithm}

\end{document}